\documentclass[10pt]{article}

\usepackage{epsfig}
\usepackage[section]{algorithm}
\usepackage{algorithmic}

\usepackage{anysize}
\usepackage{amsmath}
\usepackage{amssymb}
\marginsize{0.9in}{0.9in}{0.9in}{0.9in}

\def\qed{\hbox{\rlap{$\sqcap$}$\sqcup$}}
\newenvironment{proof}{\par\noindent{\bf Proof:}}{\mbox{}\hfill$\qed$\\}
\newtheorem{theorem}{Theorem}[section]
\newtheorem{lemma}{Lemma}[section]
\newtheorem{claim}{Claim}[section]

\newcounter{rem}
\newcommand{\ignore}[1]{ }

\begin{document}

\title{Auction Algorithm for Production Models}
\author{Junghwan Shin \thanks{Computer Science Ph.d. Candidate, IIT, Chicago, IL, e-mail: jshin7@iit.edu} and Sanjiv Kapoor \thanks{Computer Science, IIT, Chicago, IL, e-mail: kapoor@iit.edu}}
\date{2011}
\maketitle

\begin{abstract}
We show an auction-based algorithm to compute market equilibrium prices in a production model,
where consumers purchase items under separable nonlinear utility concave functions which satisfy W.G.S({\em Weak Gross Substitutes}); producers produce items with multiple linear production constraints.

Our algorithm differs from previous approaches in that the prices are allowed to both increase and decrease to handle changes in the production.
This provides a t\^{a}tonnement style algorithm which converges and provides a PTAS.
The algorithm can also be extended to arbitrary convex production regions and the Arrow-Debreu model.
The convergence is dependent on the behavior of the marginal utility of the concave function.
\end{abstract}

\section{Introduction}
The market equilibrium is a well studied problem for a general market model which includes production constraints\cite{AD54}.
Arrow and Debreu\cite{ABH59} had introduced a general production model for exchange markets and have shown proof of existence of equilibria.
In the Arrow-Debreu model, each production schedule lies in a specified
convex set. When the production model is constrained by positive production vectors only, convex programs have been obtained\cite{P93}\cite{NP83} for linear cases.
There has been considerable  recent research on the complexity of computing equilibria \cite{CV04,DPSV02,DV03,DV04,E61,GK04,J06}.

To solve the market equilibrium problem, a number of
approaches have been used including convex programming, auction-based algorithm\cite{GK04} and primal-dual\cite{DPSV02} methods.

Two techniques, the primal-dual schema and auction-based algorithms, have mainly been successful for models satisfying W.G.S.(weak gross substitutability). Other technique includes t\^{a}tonnement processes.

Jain et al.\cite{JVE05} further generalize to a model with homothetic, quasi-concave utilities which is introduced in \cite{EG59} and \cite{E61}.
However, the paper is restricted in that the concave functions are assumed to be homogeneous of degree one. We differ in that we consider functions that are {\em Weak Gross Substitutes}.
Jain et al.\cite{KK06} also give an explicit, polynomial sized convex program for the production planning model with linear utilities.
These papers utilize convex programming.
Codenotti et al. \cite{CMV05} consider gross-substitute functions with positive production constraints and provide approximation
results via the ellipsoid method.
Recent results that include production within the model include the work on price discrimination model \cite{GV10}.
Our results apply to convex production sets that are not constrained to be positive.

For the production model when there is one constraint an auction-based algorithm is provided in \cite{KMV07}.
In this paper, we give an auction-based algorithm for a
production model in which consumers have separable utilities for items
that satisfy the weak gross substitutes property.
Furthermore, producers have multiple linear production constraints.
Note that we can also consider producers that
 sell items to maximize their profit as well as purchae materials or resources to produce them.
In this model,
note that each buyer chooses a subset of items that maximizes her utility and
each producer chooses a feasible production plan that maximizes his profit at current prices.

While auction methods have been applied to market equilibrium problems, the
key aspect of many algorithms (\cite{GKV04}\cite{KMV07}) is that price discovery is
monotone since goods are always sold out, i.e. overdemanded. In this model, the
change in production plans during the course of the algorithm induces
oversupply of goods. This implies that price discovery cannot be monotone. This
is a critical difference from previous methods.
Thus unlike previous auction-based algorithms for consumer and production models, we consider auction algorithm which decrease price also.
Since each producer has an arbitrary number of production constraints,
updating prices affects production schedule.
When producers chooses a bundle of items which also
maximize consumer utilities, decreasing the prices is not required.
However, when profitable items are not demanded by consumers, price decrease may be needed.
Bounding the decrease in price can be done by small increments of the production plan along
a gradient direction specifying increasing profits. The direction is obtained
from the convex program describing the production plans. The step length is dictated
by the behavior of the utility functions, in particular on the behavior of the marginal utility function.
For simplicity of presentation our results are shown in the Fisher model with linear production constraints.
They can be extended to arbitrary convex production regions and the Arrow-Debreu model.

This paper is organized as follow:
In Section \ref{math_model}, we define the production market model and invariance of the algorithm to show the correctness and the convergence.
We describe the overall idea on the auction-based algorithm and describe in market states and market state transitions in Section \ref{top_level}.
Further, procedures are described in Section \ref{procedures} and we show invariance conditions that ensure optimality in Section \ref{invariance}.
In Section \ref{complexity}, we finally show the correctness of the algorithm in the algorithm and evaluate the complexity.
In the Appendix, we will describe the details of the algorithm.

\section{Production Model}\label{math_model}
We consider a market equilibrium problem with production, termed as {\em MEP}, with {\em nonlinear utility functions} and {\em multiple linear production constraints}.
In here, we consider nonlinear utility functions satisfy W.G.S. property.
The production model we consider has $q$ producers, along with $n$ consumers(traders) and $m$ items.
For simplicity, we assume that every item is produced in a quantity greater than or equal to $\epsilon$.

Consumer $i$ has a utility function, termed as $U_i(X_i)=\sum_ju_{ij}(x_{ij})$, and fixed initial endowment, $e_i$.
$x_{ij}$ represents the amount of allocation on item $j$ to consumer $i$,
$X_i=(x_{i1},x_{i2},...,x_{im})$ denotes the current allocation vector on the items to consumer $i$,
and $u_{ij}(x_{ij}):R^+ \rightarrow R^+$ is the function representing the
utility of item $j$ to consumer $i$.
Our assumption is that $u_{ij}$ is separable.
We denote by $v_{ij}$ the first derivative of $U_i$ w.r.t. $x_{ij}$.

We let $P$ be the vector of prices of the items, where
the $j$-th component $p_j$ represents the price of item $j$.
Then, we can represent {\em bang-per-buck}, $\forall i,j: \alpha_{ij} = v_{ij}(x_{ij})/p_j$, and let $\alpha_i$ be {\em max bang-per-buck} s.t. $\alpha_i = \max_j\alpha_{ij}$.

For producers, let $z_{sj}$ represent the quantity on item $j$ produced by producer $s$ and sold or bought by producer $s$.
Also, there are non-manufactured raw items defined as $a_j$ for all items.
Producer $s$ gains profit $\sum_jp_{j}z_{sj}$ when all items are sold out at the price $P$.
We assume that the production schedule is constrained by a set of linear inequalities.
Suppose producer $s$ has $l_s$ linear constraints.
Then, $\forall s,\ell: \sum_j z_{sj}a_{sj}^\ell \leq K_s^\ell$, where $a_{sj}^\ell$ and $K_s^\ell$ are constants determined by the production schedules by producer $s$.
$\forall s,j : z_{sj} \in \mathbb{Z}$.

Given a fixed vector of prices $P$,
the following linear programs, $CP_i(P)$ and $PP_s(P)$ represent
the optimal consumption and production schedule, respectively.

\paragraph{Consumer nonlinear programming for consumer $i$}
$$
\mbox{\bf $CP_i(P)$: } \mbox{ Maximize } \sum^m_{j=1}U_i(X_i) \mbox{ subject to}
$$
\begin{eqnarray}
\phantom{a} & \sum^m_{j=1}x_{ij}p_j \leq e_i & \forall i\\
\phantom{a} & x_{ij} \geq 0 & \forall i,j
\end{eqnarray}

\paragraph{Production linear programming for producer $s$}
$$
\mbox{\bf $PP_s(P)$: } \mbox{ Maximize } \sum^m_{j=1}p_{j}z_{sj} \mbox{ subject to}
$$
\begin{eqnarray}
\phantom{a} & \sum_j z_{sj}a_{sj}^\ell \leq K_s^\ell & \forall s,\ell\\
\phantom{a} & z_{sj} \geq 0 & \forall s,j
\end{eqnarray}

Furthermore, we assume that utilities for the items
satisfy the {\em Weak Gross Substitute} property.
Items are said to be weak gross substitutes for a buyer
iff increasing the price of any item does not decrease the buyer's demand
for other items.
Similarly, items in an economy are said to be weak  gross substitutes iff
increasing the
price of any item does not decrease the total demand of other items.

Consider the consumer maximization
problem $CP_i(P)$.
Let $S_i(P) \subset R^m_+$ be the set of optimal solutions
of the program $CP_i(P)$. Consider another price
vector $P' > P$. Items are gross
substitutes for buyer $i$ if and only if for all $X_i \in S_i(P)$ there
exists $X'_i \in S_i(P')$ such that
$p_j = p'_j \Rightarrow x_{ij} \le x'_{ij}$.

Note that since $u_{ij}$ is
concave, $v_{ij}$ is a non-increasing function. The following result
in \cite{GK04}
characterizes the class of separable concave gross substitutes utility functions.

\begin{lemma}\cite{GK04}
\label{gs-cond}
Items are gross substitutes for buyer $i$ if and only if for all $j$, $y v_{ij}(y)$ is a non-decreasing function of the scalar $y$.
\end{lemma}

\ignore{
\paragraph{Linear programs where an optimal solution consists of non-negative values}
We consider a modified LP of $PP_s(P)$, termed as $PP'_s(P)$.
Let define $z'_{sj} = z_{sj} + C$, where $C \in \mathbb{R}_+$ is a substantially big number.
$$
\mbox{\bf $PP'_s(P)$: } \mbox{ Maximize } \sum^m_{j=1}p_{j}z'_{sj} - C\sum^m_{j=1}p_{j} \mbox{ subject to}
$$
\begin{eqnarray}
\phantom{a} & \sum_j z'_{sj}a_{sj}^\ell \leq K_s^\ell + C\sum_j a_{sj}^\ell & \forall s,\ell\\
\phantom{a} & z'_{sj} \geq 0 & \forall s,j
\end{eqnarray}
We can calculate $C$ since the constraints of $PP_s(P)$ form a bounded feasible set. (@@ we need a citation here)
However, to solve $PP'_s(P)$ we also need to consider the same as $PP_s(P)$, and we provides an auction algorithm for $PP_s(P)$.
}

\paragraph{$\epsilon$-equilibrium in the Production Model}
We define an $\epsilon$-equilibrium in the production model as follows:
\begin{enumerate}
\item {\em Sold-out condition:} All produced items are sold to consumers within a factor of ($1+\epsilon$).\label{soldout}
\item {\em Bpb condition:}
The bang per buck on the items purchased by consumer $i$
should be almost the same.\label{consumerFine}
\item{\em Opt-prod condition:} The current production plan almost maximizes the profit of each producer.\label{producerFine}
\item{\em Termination condition:} Consumers spend all their endowment within a factor of $(1+\epsilon)$.\label{spendmoney}
\end{enumerate}
\noindent
The first three condition are referred to as
conditions {\em OPT} and are detailed below:
\begin{eqnarray}
& &I_1:\forall j : \sum_sz_{sj}/(1+\epsilon) \leq \sum_ix_{ij} \leq \sum_sz_{sj} \label{soldout2} \\
& &I_2:\forall i,j : x_{ij} > 0 \Rightarrow v_{ij}(x_{ij}) \leq \alpha_{ij}p_j \leq (1+\epsilon)v_{ij}(x_{ij}) \label{consumerFine2} \\
& &I_3:\forall s,j : p_j\hat{z}_{sj} \leq (1+\epsilon)p_jz_{sj} \label{producerFine2}
\end{eqnarray}
\noindent
where $\hat{z}_{sj}$ is the optimal amount of item $j$ produced by $s$.
Equilibrium is established together with the termination condition  below:
\begin{eqnarray}
& &I_4:\forall i : r_i \leq \epsilon e_i\label{spendmoney2}
\end{eqnarray}

\section{Auction Algorithm for the Production model}\label{top_level}
\subsection{Notations and Preliminaries}
We define the {\em demand set} of consumer $i$ as $D_i = \{j: v_{ij}(x_{ij})/((1+\epsilon)p_j) \leq \alpha_i \leq v_{ij}(x_{ij})/p_j\}$.
We work with a discretized price space, and at any instant
a good is sold at  at two level prices: $p_j$ and $p_j/(1+\epsilon)$.
We let $h_{ij}$ represent the amount of item $j$ that consumer $i$ buys at
price $p_j$, and $y_{ij}$ represent the amount of item $j$ that consumer $i$
purchases at price $p_j/(1+\epsilon)$.
$x_{ij}$ which is the summation of $h_{ij}$ and $y_{ij}$ represents the total quantity of allocation on item $j$ by consumer $i$.
We also let $r_i$ denote {\em residual money} which can be
calculated as $e_i-\sum_j(p_jh_{ij}+p_jy_{ij}/(1+\epsilon))$.

We denote by $\mathcal{Z}_s$ a set of {\em feasible production plan(schedule)} of producer $s$
 i.e.,
$\mathcal{Z}_s = \{Z : A^TZ \leq K_s\}$, where $A$ represents a $m \times \ell$ matrix whose row and column correspond to items and constraints, respectively. $Z$ represents a production plan and $K_s$ represents capacity  constraints, $Z,K_s \in R^m$.
Note that $Z$ is an unconstrained vector, that is, if a producer produces item $j$
as a product, then $Z$ may be positive; if a producer consumes item $j$ as a material
or resource, then $Z$ may be negative.
Let us define a {\em profitable} production plan of producer $s$ as $\hat{Z}_s = \arg\max_{Z \in \mathcal{Z}_s}\{P^TZ | A^TZ \leq K_s\}$ at price $P$.

We define a set of {\em over-demanded} items $\mathcal{O}_d = \{j: \exists i$ s.t. $\sum_sz_{sj} < \sum_ix_{ij}\}$.
Moreover, we define a set of {\em over-supplied} items as $\mathcal{O}_s = \{j: \sum_sz_{sj} > \sum_ix_{ij}\}$.

\subsection{Algorithm Overview}
We apply the auction mechanism to solve the problem of finding equilibrium
in the linear production model.
Note that in this application we allow the algorithm to
decrease and increase prices.
The market equilibrium in the production model is to find a price vector at which
consumers maximize their utility and producers maximize their profit.
Consumers want to purchase items at lower price, while producers want to sell items at
higher price, creating complementary optimization constraints.
In fact, the producers would wish to
produce as much demand as the market can sustain, to maximize their profit.
The algorithm we design allows producers to increase supply incrementally to change their production plan according to the current market price.

The algorithm is iterative and its progress can be measured by the changes in its state during successive steps.
The algorithm we design allows producers to increase supply incrementally to change their production plan according to the current market price.
To prove convergence we will show that the producer's profit increases during phases
of the algorithm.

\paragraph{Overview  of the algorithm}
The algorithm starts with an initialization , procedure {\em initialize} where  the
price vector is initialized, each producer is
assigned an initial production schedule and each consumer is assigned goods, all quantities
small enough.
In procedure {\em algorithm\_main}, the algorithm determines if consumers have no extra demand and
producers have no way to increase their profit and consequently stops.
The invariant conditions to be
maintained at the end of every phase are
conditions (\ref{soldout2}),(\ref{consumerFine2}) and (\ref{producerFine2}).
If at the beginning of a phase no consumer has any residual money left, the procedure
can terminate satisfying all the conditions.
Otherwise, it is determined if the consumer can outbid other consumers to acquire her
desired items in procedure {\em satisfy\_demand} and subsequently adjust
the allocations of other consumers {\em adjust\_bpb} to ensure condition (\ref{consumerFine2}).
Each of these procedures may result in an increase in price.
Consequently, the production schedule may required to be changed. This is
done in  procedure {\em prod\_reschedule}.

\ignore{
At the beginning of the algorithm, procedure {\em initialize} is called to initialize all variables.
In procedure {\em algorithm\_main}, the algorithm stops and converges if consumers have no extra demand and producers have no way to increase their profit.
Otherwise, consumer checks if the consumer can outbid other consumers to acquire her desired items in procedure {\em satisfy\_demand}, and {\em adjust\_bpb} allows consumers to maintain bang-per-buck condition according to $\alpha_{ij}$.
After the consumer tries to outbid, a current price vector may change, and producers may update their production schedule in procedure {\em prod\_reschedule}.
}
\begin{algorithm*}[!h]
\caption{algorithm\_main}
\begin{algorithmic}[1]
\STATE {\em initialize}
\STATE \hspace{2 cm} // bpb, opt-prod and sold-out conditions hold.
\WHILE{there exists consumer $i$ with residual money}
\STATE update bang-per-buck
\STATE \hspace{1.6 cm} // bpb, opt-prod and sold-out conditions hold.
\STATE {\em satisfy\_demand}($i$)
\STATE \hspace{1.6 cm} // sold-out condition holds.
\STATE {\em adjust\_bpb}
\STATE \hspace{1.6 cm} // bpb and sold-out conditions hold.
\STATE {\em prod\_reschedule}
\STATE \hspace{1.6 cm} // bpb, opt-prod and sold-out conditions hold.
\ENDWHILE
\STATE \hspace{2 cm} // termination, bpb, opt-prod and sold-out conditions hold.
\end{algorithmic}
\end{algorithm*}

\begin{algorithm}[!h]
\caption{satisfy\_demand($i$)}
consumer $i$ with extra demand purchase items
\begin{algorithmic}[1]
\STATE update a set of demanded items of consumer $i$
\IF{item $j$ is available at lower price from consumer $k$}
\STATE {\em outbid}($i,k,j,\alpha_{ij}$)
\STATE update bang-per-buck
\ELSE
\STATE {\em raise\_price}($j$)
\ENDIF
\end{algorithmic}
\end{algorithm}

\begin{algorithm}[!h]
\caption{adjust\_bpb}
bang-per-buck condition holds by outbid
\begin{algorithmic}[1]
\WHILE{bang-per-buck not hold, i.e $\exists i: r_i > 0$ and $\exists j: \alpha_{ij}p_j < v_{ij}(x_{ij})$}
\IF{item $j$ is available at lower price from consumer $k$}
\STATE {\em outbid}($i,k,j,\alpha_{ij}$)
\ELSE
\STATE {\em raise\_price}($j$)
\ENDIF
\ENDWHILE
\end{algorithmic}
\end{algorithm}


The procedure, {\em prod\_reschedule}, determines a feasible direction that improves
the profit of a producer, by using a linear program.
A small step of length $\delta$ is taken along the feasible direction. The step length is chosen so
that optimality conditions hold as shown later.
The step provides a new production plan
based on the current price vector $P$, and the production of items may decrease and increase as
compared with the previous production plan.
Also, note that $z_{sj}$ may be negative in that producer $s$ consumes item $j$ to produce other items;
The change in the production of items leads to goods that are either oversupplied or
over demanded.
The demand and supply are then balanced in the procedures: {\em bal\_od} and {bal\_os}.

Due to a change in the production plans, a good $j$ may be reduced in quantity.
Then the consumers holding the good $h$ have reduced allocation and have money $m$,
returned to them.
To ensure their bang-per-buck conditions,
consumers outbid other consumers in procedure {\em bal\_od}.
After procedure {\em bal\_od}, procedure {\em adjust\_bpb} is called to
ensure that bang-per-buck conditions are met. Note that this may involve
a rise in price.

A change in production may also involve increased production of
some particular good.
In procedure {\em bal\_os}, the algorithm will balance over-supplied items
so that there is no over-supplied item in the market.

Procedure {\em bal\_os} balances the demand and supply of
over-supplied items by four different procedures : {\em purchase\_money}, {\em transfer\_money}, {\em sell\_lprice} and {\em decrease\_price}.
Procedure {\em purchase\_money} is allowed for consumer with extra demand to acquire items.
Unless consumer has residual money, procedure {\em transfer\_money} will take money from other items to purchase her demand items.
If the market has an item that is over-supplied, then producers will provide
items at a lower price as shown in procedure {\em sell\_lprice}.
Finally, Procedure {\em decrease\_price} will be called if there is no available item at a higher price.


\ignore{
In procedure {\em shrink\_model}, we consider a bipartite graph $(U,V,E)$ where $U$ is a set of producers and consumers; $V$ is a set of items.
$E$ is a set of edges associated with demand and production.
If item $j$ is demanded by producer $s$ (consumer $i$), then there is an edge from $s ($ consumer $i)$ to $j$.
Conversely, item $j$ is produced by producer $s$, then there is an edge from $j$ to $s$.
We check whether a disjoint graph $(U',V',E')$ s.t. consumers have no demand on $V' \subseteq \mathcal{O}^-_d$, and consider a new graph $(U \setminus U', V \setminus V', E \setminus E')$.
Now, we consider a market model with producers and consumers in $U \setminus U'$ and items in $V \setminus V'$.
}

\begin{algorithm*}[!h]
\caption{prod\_reschedule}
\begin{algorithmic}[1]
\WHILE {execute {\em lp\_solver} and if not at optimality, $\forall s,j : z'_{sj} := z_{sj} \pm \epsilon'\sum_sz_{sj}/q$}
\IF {producer can increase profit}
\STATE {\em bal\_od}($j$) : $\forall j \in \mathcal{O}_d$
\STATE {\em adjust\_bpb}
\STATE \hspace{1.6 cm} // bpb condition holds.
\STATE {\em bal\_os} : $\forall j \in \mathcal{O}_s$
\STATE \hspace{1.6 cm} // bpb and sold-out conditions hold.
\STATE {\em check\_profit}
\ELSE
\STATE break
\ENDIF
\ENDWHILE
\STATE \hspace{2 cm} // bpb, sold-out and opt-prod condition hold.
\end{algorithmic}
\end{algorithm*}

\ignore{
Note that procedure {\em shrink\_model} deletes some producers and items that are not relevant to the current process.
However, the procedure only affects to the process temporarily not permanently.
That is, procedure {\em unshrink\_model} recovers and all producers and all items are considered.
}

Finally, the algorithm checks whether there is
total profit increase in the overall process at each iteration.
Procedure {\em check\_profit} will be called to check whether consumers
spend more than they did at the prior iteration.

If the total profit does not increase, then changes of productions plans
are voided. It is shown later that in this case the production plans
are at near optimality. The algorithms then exits {\em prod\_reschedule}.

\paragraph{Choosing a value of $\epsilon'$}
We will assume that the producers will change their production plan
by a factor of $1+\epsilon'$.
We choose $\epsilon'$ to ensure an approximation of $(1+\epsilon$.

We choose constants as follows:
set $\epsilon_1$ s.t. $\forall i,j: v_{ij}(x_{ij})/(1+\epsilon) \leq v_{ij}((1+\epsilon_1)x_{ij}) \leq v_{ij}(x_{ij})$ and
$\forall i,j: v_{ij}(x_{ij}) \leq v_{ij}((1-\epsilon_1)x_{ij}) \leq (1+\epsilon)v_{ij}(x_{ij})$.
Let $\epsilon_2 = \epsilon^3/e, e = \sum_ie_i$ and let $\epsilon' = \min(\epsilon_1,\epsilon_2)$.

\subsection{States of the market}
The state of the market at time instant $t$ is represented
as a 3-tuple $(P(t),Z(t),X(t))$, where $P(t),Z(t)$ and $X(t)$ represent the  price vector,
the production plan and the current allocation (or demand) at time $t$, respectively.
We will also use $p_j(t), z_{sj}(t)$ and $x_{ij}(t)$ to denote
the price of item $j$, the production plan for item $j$ by producer $s$ and
the demand of consumer $i$ for item $j$ at time $t$, respectively.

We consider the transitions of the state of the market at time $t$, $(P(t),Z(t),X(t))$,
to the states at time $t+1$, $(P(t+1),Z(t+1),X(t+1))$ via an algorithmic process.

\paragraph{State Transitions of the market}
We consider the state transitions :
\begin{eqnarray}
P(t+1) &= f_P(P(t),Z(t),X(t))\\
X(t+1) &= f_X(P(t),Z(t),X(t))\\
Z(t+1) &= f_Z(P(t),Z(t),X(t))
\end{eqnarray}
where we design the functions $f_P, f_X$ and $f_Z$ to satisfy conditions (\ref{soldout2}) through (\ref{spendmoney2})on the state of the market.
The state transition function $f_P, f_X$ and $f_Z$ are computed in algorithm\_main.
At each iteration conditions (\ref{soldout2}) through (\ref{producerFine2}) should be satisfied.
Furthermore we will show that at termination, condition (\ref{spendmoney2}) is true and thus at termination the market is in a state of
equilibrium.

The function $f_P$ depends on a parameter $P(t),Z(t)$ and $X(t)$.
In other words, $p_j(t)$ will increase by a factor of $(1+\epsilon)$ unless $\exists i : y_{ij}(t) > 0$.
The procedure outbid may occur on either consumer's side or producer's side.
Consumer $i$ outbids consumer $k$ to acquire item $j$, or producer $s$ outbids consumer $k$ to consume item $j$ to produce item $j'$.

To compute $f_X$ we invoke procedure satisfy\_demand which allows consumer $i$ with $r_i > \epsilon e_i$ to purchase item $j \in D_i$ at $p_j(t)$ by acquiring item $j$ from consumer $k$.
If the demand and the supply does not match, then procedures outbid, bal\_os or bal\_od may be called so that $X(t+1)$ may change.

To update the production schedule and compute $f_Z$, we invoke procedure prod\_reschedule.
According to $P(t)$, a linear program solver returns an optimal production schedule.
As mentioned before, a next production schedule is determined by the demand and the supply.

\subsection{Market State Transitions}
\paragraph{States of market in procedure satisfy\_demand}
We now detail the change of the states of market as well as maintenance of
invariants in procedure satisfy\_demand inside procedure algorithm\_main.
In a market state, $(P(t),Z(t),X(t))$, procedure satisfy\_demand calls either
outbid or raise\_price when $\forall j: \alpha_{ij}p_j = v_{ij}(x_{ij})$.
Since there exists consumer $i$ s.t. $r_i > \epsilon e_i$,
she will buy item $k = \arg\max_j\alpha_{ij}$.

Note that before procedure satisfy\_demand,
$\forall j: \alpha_{ij} = v_{ij}(x_{ij}(t))/p_j(t)$.
Then, for item $j \neq k$, $\alpha_{ij}p_j(t) = v_{ij}(x_{ij}(t))$
and $\alpha_{ij}p_j(t+1) = v_{ij}(x_{ij}(t+1))$.
In the case of item $k$, $\alpha_{ik}p_k(t) = v_{ik}(x_{ik}(t))$ and $\alpha_{ik}p_k(t+1) = (1+\epsilon)v_{ik}(x_{ik}(t+1))$.
However, $\alpha_{ik}$ will be updated, and finally, $\forall j: \alpha_{ij}p_j(t+1) = v_{ij}(x_{ij}(t+1))$.

Since there is, $P(t)$ and $X(t)$ may be not equal $P(t+1)$ and $X(t+1)$, respectively. However, no change in production plan, $Z(t) = Z(t+1)$.

\paragraph{States of market in procedure adjust\_bpb}
We next consider changes to the states of the market and
the invariants in procedure adjust\_bpb inside procedure algorithm\_main.
In a market state, $(P(t),Z(t),X(t))$, procedure adjust\_bpb calls either outbid or raise\_price when $\exists j: \alpha_{ij}p_j < v_{ij}(x_{ij})$.
For the corresponding consumer $i$, $\forall j: \alpha_{ij}p_j = v_{ij}(x_{ij})$ at the end of adjust\_bpb.

For item $j$ s.t. $\alpha_{ij}p_j(t) < v_{ij}(x_{ij}(t))$, procedure outbid makes $p_j(t) = p_j(t+1), x_{ij}(t) < x_{ij}(t+1)$; on the other hand, procedure raise\_price makes $(1+\epsilon)p_j(t) = p_j(t+1)$. The aim of this procedure is $\alpha_{ij}p_j(t+1) = v_{ij}(x_{ij}(t+1))$.

Therefore, $P(t)$ and $X(t)$ may be not equal to $P(t+1)$ and $X(t+1)$, respectively. Since there is no change in production plan, $Z(t) = Z(t+1)$.

\paragraph{States of market in procedure prod\_reschedule}
Given $(P(t),Z(t),X(t))$ (the current market state) procedure prod\_reschedule
calls lp\_solver which returns optimal production plans, $\hat{Z}_s$ according to $P(t)$ (as shown in line 2 in procedure prod\_reschedule).
$Z_s(t+1)$ depends on $Z_s(t)$ and $\hat{Z}_s$ as follows:

$$\forall s,j : z_{sj}(t+1) =
\left\{
\begin{array}{cl} \mbox{$z_{sj}(t) + \Delta$, $\Delta = \epsilon' \sum_sz_{sj}(t)/q$} & \mbox{if $\|\Delta\| \leq ||\hat{z}_{sj}(t)-z_{sj}(t)||$} \\
\mbox{$\hat{z}_{sj}(t)$} & \mbox{otherwise}\end{array}\right.
$$

Since the production plan changes, allocations of consumers and production of producers might change.
These changes may increase or decrease price of items.
We describe more details in Section \ref{procedures}.


\section{Other procedures in the algorithm}\label{procedures}
In this section we will discuss other important and  their properties.


\begin{enumerate}
\item outbid($i,j,k,\alpha$) :
consumer $i$ with surplus will outbid consumer $k$ to acquire item $j$.
The quantity that is outbid of item $j$ is determined by the utility function and the current allocation so as to maintain the invariance.
After we set the amount outbidden, we need to update the allocation and the surplus of both consumers $i$ and $k$.

\item purchase\_money($i,j,t_{oversupply}$) :
for item $j \in \mathcal{O}_s$, the procedure will check whether item $j$ is demanded by consumer $i$ with surplus.
If exists, consumer $i$ will purchase item $j$ as much as the minimum of the quantity oversupplied and extra demand.
Here we also need to update the allocation of consumer $i$ on item $j$ and the surplus.

\item transfer\_money($i,j,t_{oversupply}$):
in the case when there is no consumer with surplus for item $j \in \mathcal{O}_s$, the procedure transfer\_money will be called.
That is, consumer $i$ with demand on item $j$ will give up item $j'$ and spend the returning money to buy item $j$.
The allocation of consumer $i$ will be updated for item $j$ and $j'$.

\item sell\_lprice($i,j,t_{oversupply}$) :
if there is consumer neither with surplus nor with items transferable, producers will offer item $j$ at lower price.
Remember that the algorithm accepts two-level price: a higher price and a lower price by a factor of ($1+\epsilon$).
Producers are forced to supply item $j$ at a lower price, then either $j$ is fulfilled or no consumer has item $j$ at a higher price.

\item decrease\_price($j$) :
when procedure sell\_lprice fails to fulfill item $j$, we assure that the allocation of item $j$ is at a lower price to any consumer.
Then, procedure decrease\_price will be called, and the while statement of procedure bal\_os will check whether procedures purchase\_moeny, transfer\_money or sell\_lprice might be called.

\item raise\_price($j$) :
after the procedure prod\_reschedule, if there is a consumer with surplus then consumer $i$ outbids other consumers.
If consumer $i$ is not satisfied then she will increase the price of item $j$ by a factor of ($1+\epsilon$).
Procedure raise\_price itself does not change the allocation, but affects invariance that should be maintained in the algorithm.
As consumers acquire item $j$ at higher price than previous, the producers may
change their production plan. This will be checked in the program.

\end{enumerate}

\section{Invariances}\label{invariance}
In this section we will prove the invariances that hold during the course of
the algorithm.
We first discuss the properties of two procedures.
\subsection{Procedure raise\_price}  We  consider the number of occurrences of  procedure raise\_price.

\begin{lemma}\label{rplemma}
In procedure prod\_reschedule at most one occurrence of procedure raise\_price is required.
\end{lemma}

\begin{proof}
We claim that procedure raise\_price occurs at most once per each item at each iteration in procedure prod\_reschedule.

Definitely, procedure bal\_os will not call procedure raise\_price.
Procedure prod\_outbid allows producers to outbid consumers , but procedure raise\_price will not be called.

Procedure bal\_od will be called , and consumer $i$ may potentially violate condition (\ref{consumerFine2}) because $v_{ij}(x_{ij})$ increases when $x_{ij}$ decreases. To adjust bang-per-buck condition, procedure adjust\_bpb will be called.

If $\forall k: y_{kj} = 0$ and $\alpha_{ij}p_j < v_{ij}(x_{ij})$, then procedure raise\_price will be called once.
Note that $\alpha_{ij}p_j = v_{ij}(x_{ij}) \leq v_{ij}(x'_{ij}) \leq (1+\epsilon)\alpha_{ij}p_j$, where $x'_{ij}$ corresponds to the reduced amount of item $j$ of consumer $i$.
Since $\epsilon'$ is chosen to satisfy $v_{ij}(x_{ij}) \leq v_{ij}((1-\epsilon')x_{ij}) \leq (1+\epsilon)v_{ij}(x_{ij})$,
$\alpha_{ij}p_j \leq v_{ij}(x'_{ij}) \leq (1+\epsilon)\alpha_{ij}p_j$.
After raise\_price, i.e. $p'_j = (1+\epsilon)p_j, \alpha_{ij}p'_j \geq v_{ij}(x'_{ij})$.

Therefore, at most one occurrence of procedure raise\_price will be enough in procedure prod\_reschedule.
\end{proof}

\subsection{Procedure decrease\_price}
Similarly, we  consider the number of occurrences of  procedure decrease\_price.
\begin{lemma}\label{dplemma}
Procedure decrease\_price of each item occurs at most once per each iteration in procedure prod\_reschedule.
\end{lemma}

\begin{proof} Let $p$ and $p'$ denote a previous price vector and a current price vector, respectively.
Let $z$ and $z'$ denote a previous production plan and a current production plan, respectively.

Suppose, for contradiction, that procedure decrease\_price on item $j$ occurs at least twice in an iteration.
Remember that we have two additional variables $\epsilon_1$ and $\epsilon_2$ and $\epsilon' = \min(\epsilon_1,\epsilon_2)$ in a paragraph in Section \ref{top_level}.
It is enough to show the case when $\epsilon' = \epsilon_2 = \epsilon^3/\sum_ie_i$.

The money required to consume the oversupplied items is
\[\epsilon'\sum_s\sum_{j'}z_{sj'}p_{j'} = \epsilon^3\sum_s\sum_{j'}z_{sj'}p_{j'}/\sum_ie_i\]

Calling decrease\_price on item $j$ returns the following amount of money to all consumers,
\[((\epsilon^2+2\epsilon)\sum_ih_{ij} + \epsilon\sum_iy_{ij})p_j\]

For decrease\_price to occur more than once, the required money must be greater than the money returned, i.e.
\[\frac{\epsilon^3\sum_s\sum_{j'}z_{sj'}p_{j'}}{\sum_ie_i} > ((\epsilon^2+2\epsilon)\sum_ih_{ij} + \epsilon\sum_iy_{ij})p_j\]
\noindent
Since $\sum_s\sum_{j'}z_{sj'}p_{j'} \leq \sum_ie_i$,
\begin{eqnarray*}
& &\epsilon^3 > ((\epsilon^2+2\epsilon)\sum_ih_{ij} + \epsilon\sum_iy_{ij})p_j\\
& &\Rightarrow \epsilon^2 > ((\epsilon+2)\sum_ih_{ij} + \sum_iy_{ij})p_j > \sum_ix_{ij}p_j\\
& &\Rightarrow\epsilon^2 > \sum_ix_{ij}p_j \geq \epsilon^2
\end{eqnarray*}
\noindent
Since every item is demanded by at least one consumer,
the money spent on item $j$ , $\sum_ix_{ij}p_j$ should be equal or at least $\epsilon^2$, i.e. $\sum_ix_{ij}p_j \geq \epsilon^2$.
This is a contradiction.
\end{proof}

\paragraph{Ensuring sold-out condition (\ref{soldout2}):}
Procedures adjust\_bpb and satisfy\_demand will not violate the invariance because there is no change in production plan and the amount of items sold does not decrease.

When producer $s$ changes a production plan from $Z_s(t)$ to $Z_s(t+1)$, the following cases arise.
\begin{itemize}
\item $\sum_sz_{sj}(t+1) - \sum_ix_{ij}(t+1) > 0$
\item $\sum_sz_{sj}(t+1) - \sum_ix_{ij}(t+1) < 0$
\end{itemize}
Procedure bal\_od resolves the first case, and the second case can be resolved in procedure bal\_os.
That is, procedure bal\_od balances all items in $\mathcal{O}_d$; procedure bal\_os balances oversupplied items.

\begin{enumerate}
\item For $j \in \mathcal{O}_d$, if $\exists i: x_{ij} > 0$ then $x_{ij}(t+1)= x_{ij}(t) - \min(x_{ij}(t) , \xi)$
shown in procedure bal\_od.

\item For $j \in \mathcal{O}_s$, if $\exists i: j \in D_i$ and $r_i > \epsilon e_i$, then
consumer $i$ will purchase item $j$ by using procedure purchase\_money (as shown in line 2 and 3 in procedure bal\_os).
It ensures that $p_j(t) = p_j(t+1)$ and while $\sum_sz_{sj}(t+1) = \sum_ix_{ij}(t+1)$, OPT(conditions sold-out(\ref{soldout2}) through bang-per-buck(\ref{consumerFine2})) conditions are met.

\item Otherwise, i.e. $\exists i$ s.t. $j \in D_i \bigcap \mathcal{O}_s$ and $r_i \leq \epsilon e_i$
\begin{enumerate}
\item To ensure that item $j$ is sold out, the item is sold to consumer $i$
s.t. $j \in D_i$.  This can only happen if consumer $i$ transfers money
from a item that provides lower bang-per-buck, i.e consumer $i$
spends money to purchase item $j$ instead of
item $k$ s.t. $x_{ik} > 0 \bigwedge \alpha_{ik} < \alpha_i$(in procedure transfer\_money).
The procedure transfer\_money ensures that sold-out(\ref{soldout2}) condition is satisfied.

\item If item $j$ is still over-supplied, then producers may offer item $j$ at a
lower-level price (remember there are 2 price levels).
Note that the current price $p_j$ does not change,
but consumers will acquire item $j$ more with the same amount of money.
When $\sum_sz_{sj}(t+1) - \sum_ix_{ij}(t+1) \leq \epsilon\sum_ih_{ij}(t+1)$, item $j$ will be sold out with the same money(in procedure sell\_lprice).
Then, $p_j(t) = p_j(t+1)$ and $\sum_ix_{ij}(t) < \sum_ix_{ij}(t+1)$.
Conditions OPT will be satisfied.

\item Item $j$ will not be sold out yet if $\sum_sz_{sj}(t+1) - \sum_ix_{ij}(t+1) > \epsilon\sum_ih_{ij}(t+1)$. Then, instead of providing item $j$ at lower-level price, producers will offer item $j$ at a lower price than before i.e. $p_j(t+1) = p_j(t)/(1+\epsilon)$(in procedure decrease\_price).
Then, $p_j(t) > p_j(t+1)$ and $\sum_ix_{ij}(t) < \sum_ix_{ij}(t+1)$.
If still $\sum_sz_{sj}(t+1) > \sum_ix_{ij}(t+1)$, then conditions OPT
may not be satisfied. However, then the iterations within the procedure
bal\_os will repeat until $\sum_sz_{sj}(t+1) = \sum_ix_{ij}(t+1)$.
\end{enumerate}
\end{enumerate}

\begin{lemma}\label{sold_out_prod}
Condition (\ref{soldout2}) is satisfied at the end of procedure prod\_reschedule.
\end{lemma}

\paragraph{Ensuring bpb condition (\ref{consumerFine2}):}

\begin{lemma}\label{main_conbpb}
Condition (\ref{consumerFine2}) is satisfied after procedures satisfy\_demand and adjust\_bpb.
\end{lemma}

\begin{proof}
In initialization, $v_{ij}(x_{ij}) \leq \alpha_{ij}p_j$ is true. When procedure adjust\_bpb occurs in procedure algorithm\_main, $v_{ij}(x_{ij}) \leq \alpha_{ij}p_j$ is always true since adjust\_bpb iterates until $\forall j: v_{ij}(x_{ij}) = \alpha_{ij}p_j$.
Also, $\alpha_{ij}p_j \leq (1+\epsilon)v_{ij}(x_{ij})$ is true.

In procedure satisfy\_demand, note that $\alpha_{ij}$ is updated according to $v_{ij}(x_{ij})$ and $p_j$ which implies that $\forall j: v_{ij}(x_{ij}) = \alpha_{ij}p_j$. For $k = \arg\max_j\alpha_{ij}$, consumer $i$ purchases $x'_{ik}$ amount of item $k$ s.t. $v_{ik}(x'_{ik}) = \alpha_{ik}p_k/(1+\epsilon)$. Then, $v_{ij}(x_{ij}) \leq \alpha_{ij}p_j$ is still true.
\end{proof}

\begin{lemma}\label{bpb_od}
Condition (\ref{consumerFine2}) is satisfied at the end of procedure bal\_od.
\end{lemma}

\begin{proof}
We claim that condition (\ref{consumerFine2}) holds at the end of bal\_od based on Lemma \ref{rplemma}.

Let us consider procedure bal\_od where consumers balance their over-demand items.
Let consumer $i$ give up some amount on item $j$.
Consumer $i$ may potentially violate condition (\ref{consumerFine2}) because $v_{ij}(x_{ij})$ increases when $x_{ij}$ decreases.

However, as shown in Lemma \ref{rplemma}, raise\_price occurs at most once at each iteration.
Therefore, any price or allocation change does not violate condition (\ref{consumerFine2}).
\end{proof}

\begin{lemma}\label{bpb_os}
Condition bpb (\ref{consumerFine2}) is satisfied at the end of procedure bal\_os.
\end{lemma}
\begin{proof}
Let producer $s$ change his production plan, and $j \in \mathcal{O}_s$.
There are two possible situations in procedure bal\_os: 1)procedure decrease\_price is not called. 2)procedure decrease\_price is executed.

Procedure decrease\_price is not called since consumers consume all amount of item $j$.
Let consumer $i$ have surplus, i.e. $r_i > \epsilon e_i$, and $j \in D_i$.
Since consumer $i$ has surplus, she will buy item $j$ as shown in procedure purchase\_money.
However, it will not violate condition (\ref{consumerFine2}) because consumers will only purchase item $j$ within a factor of ($1+\epsilon$) of bang-per-buck.
That is, previously $\alpha_{ij}p_j = v_{ij}(x_{ij})$, and $x'_{ij}$ increases such that $v_{ij}(x_{ij}) \leq v_{ij}(x'_{ij}) \leq (1+\epsilon)v_{ij}(x_{ij})$.

Although consumer $i$ has surplus, the algorithm is allowed for consumers to purchase items by calling procedures transfer\_money and sell\_lprice. Since the amount of production change is well defined, both procedures do not violate bpb condition.
Recall that procedure transfer\_money allows consumer $i$ to buy item $j$ to balance bang-per-buck.

When all other procedures do not resolve the over-demand of item $j$, procedure decrease\_price occurs.
The occurrence of procedure decrease\_price implies that consumers who want to buy item $j$ do not have enough money.
After procedure decrease\_price occurs, let consumer $i$ buy item $j$, and let $x'_{ij}$ correspond to the new amount of item $j$ of consumer $i$.

Previously $\alpha_{ij}p_j = v_{ij}(x_{ij})$ and after one occurrence of procedure decrease\_price as shown in Lemma \ref{dplemma},
\[(1+\epsilon)\alpha_{ij} = v_{ij}(x_{ij})/p'_j\]
Consumer $i$ now outbids to take item $j$. Then,
\begin{eqnarray*}
& &v_{ij}(x_{ij})/(1+\epsilon)p'_j \leq v_{ij}(x'_{ij})/p'_j \leq v_{ij}(x_{ij})/p'_j\\
& &(1+\epsilon)\alpha_{ij}/(1+\epsilon) \leq v_{ij}(x'_{ij})/p'_j \leq (1+\epsilon)\alpha_{ij}\\
& &\alpha_{ij} \leq v_{ij}(x'_{ij})/p'_j \leq (1+\epsilon)\alpha_{ij}
\end{eqnarray*}
\end{proof}

\paragraph{Ensuring opt-prod condition (\ref{producerFine2}):}

\ignore{
\begin{claim}\label{tpclaim}
During procedure prod\_reschedule, total profit gained by producers increases by a factor of ($1+\epsilon'$) while profit of some producers may decrease by a factor of $(1+\epsilon)$.
\end{claim}
\begin{proof}
Let denote by $V_s(p,z)$ the current profit of producer $s$ at price $P$ and production $z$.
Let denote by $V'_s(p,z)$ the next profit of producer $s$ at price $P$ and production $z$.
There are four procedures that may potentially violate $V'_s(p,z) \geq V_s(p,z)/(1+\epsilon)$.
We claim that $V'_s(p,z) \geq V_s(p,z)/(1+\epsilon)$.

In bal\_od, consumers outbid other consumers to acquire items resulting in the potential violation of profit increase.
However, consumers will spend their returning money for their demanded items later.

Since each consumer will outbid with a budget equal to at most the money returned, it is guaranteed that raise\_price occurs at most once as shown in Lemma \ref{rplemma}.
Therefore, procedure producer\_outibd may reduce profit of producer as much as $\epsilon V_s(p,z)/(1+\epsilon)$ at most.

Now, let us consider total profit gained by producers at one iteration.
If total profit increases less than $(1+\epsilon')\min_sV_s(p,z)$, prod\_roll\_back is called.
Otherwise, the total profit of producers will increase at least as much as $(1+\epsilon')\min_sV_s(p,z)$.
Therefore, total profit will increase at least by a factor of ($1+\epsilon'$).
\end{proof}
}

It happens that no producer reschedule their production plan.
When there is no profit for producers, producers do not want to change their own production schedule.
Even in the case, we show that producers will satisfy the following inequality as follow:
\[\forall s,j : p_jz_{sj} \leq p_jz^*_{sj} \leq (1+\epsilon)^2p_jz_{sj}\label{producerFine3}\]

\begin{lemma}
\label{no-change-good}
When the procedure roll\_back occurs, producer $s$ still has $O(1+\epsilon)$-approximation optimal profit according to the current prices.
\end{lemma}

\begin{proof}
Let $V(z,p)$ denote the current profit of producer $s$ on production plan $z$ according to the price vector $p$. Similarly, let $V(z',p)$ be the next profit of producer $s$ on production plan $z'$ according to the price vector $p$. $V(\bar{z},p)$ denotes the profit of producer $s$ when he has the production plan $\bar{z}$, the optimal production plan according to the price vector $p$.

When the production plan shifts from $V(z,p)$ to $V(z',p)$, price may change to $p'$ due to procedure decrease\_price.
The occurrence of procedure decrease\_price may violate that producer $s$ increases his profit.
In procedure prod\_reschedule, if producer $s$ has no profit increase, then there is no production change.
We will show that when procedure roll\_back, producer $s$ still guarantees his production profit is well bounded compared with the optimal production profit.
We show this by proving that
\[V(\bar{z},p) \leq (1+2\epsilon)V(z,p)\]

Note that $z,z',\bar{z}$ are points on the poly-tope of multiple linear production constraints.
By the property,
\[\bar{z} = z + (z' - z)/\sigma \Rightarrow \bar{z}p' = zp' + p'(z' - z)/\sigma\]
where $\sigma = |z'-z|/|\bar{z}-z|$.
\[V(\bar{z},p') = V(z,p')+(V(z',p')-V(z,p'))/\sigma\]
Let $V(z',p') = (1+\epsilon')V(z,p')$,
\[V(\bar{z},p') = V(z,p')+\epsilon' V(z,p')/\sigma\]
Let $\epsilon'/\sigma \leq \epsilon$,
\[V(\bar{z},p') \leq (1+\epsilon)V(z,p)\]
Note that the procedure decrease\_price occurs at most once in procedure prod\_reschedule which implies that $\forall j: p_j \leq (1+\epsilon)p'_j$.
\[V(\bar{z},p) \leq (1+\epsilon)V(\bar{z},p') \leq (1+2\epsilon)V(z,p)\]
\end{proof}

\section{Analysis of the algorithm}\label{complexity}
\subsection{Invariances}
\begin{lemma}\label{all-opt}
When the algorithm terminates conditions OPT are satisfied and the algorithm returns a $O(1+\epsilon)$-approximate optimum.
\end{lemma}

\begin{proof}
Condition sold-out (\ref{soldout2}) is true at the end of each iteration of the algorithm as shown in Lemma \ref{sold_out_prod}, condition (\ref{soldout2}) holds at the end of prod\_reschedule.

Condition bang-per-buck (\ref{consumerFine2}) is satisfied at the end of algorithm as shown in Lemma \ref{main_conbpb}, bpb condition holds after calling procedures adjust\_bpb and satisfy\_demand. In the case of procedure prod\_reschedule, as shown in Lemma \ref{main_conbpb}, Lemma \ref{bpb_od} and Lemma \ref{bpb_os}, condition (\ref{consumerFine2}) holds.

Condition (\ref{producerFine2}) is true at the end of prod\_reschedule as shown \ref{no-change-good}.
\end{proof}

\subsection{Convergence}
Now, let us consider time complexity which guarantees that our algorithm converges in
We show the convergence of each procedure before we provide time complexity.

Remember that bidding is organized in rounds.
In each round every consumer is picked once and reduces his surplus until $r_i = 0$.
If there is no outbid, then procedure raise\_price will occur followed by procedure prod\_reschedule.
In procedure prod\_reschedule, the algorithm will balance between the demand and the production.

Let $N_0 = \displaystyle\log_{1+\epsilon}\frac{e}{\epsilon e_{min}}$, $N_1 = \displaystyle\log_{1+\epsilon}\frac{p_{max}}{p_{min}} = \log_{1+\epsilon}\frac{e}{\epsilon}$ and $N_2 = \displaystyle\log_{1+\epsilon'}\frac{e}{\epsilon^3e_{min}}$, where $p_{max} = \max_jp_j$ and $p_{min} = \min_jp_j$.

\begin{claim}\label{claim2}
After $N_1*N_2$ rounds of bidding, either the algorithm terminates or a round robin completes.
\end{claim}

\begin{proof}
If any producer does not reschedule, then price rises.
In the worst case, the maximum number of calls to raise\_price is bounded by $\displaystyle\log_{1+\epsilon}\frac{p_{max}}{p_{min}}$, termed $N_1$.

If producers gain profit, then at least one producer will increase her profit by a factor of ($1+\epsilon'$) of the previous profit.
We can bound the number of occurrence as $\displaystyle\log_{1+\epsilon'}\frac{e}{\epsilon^3e_{min}}$, termed $N_2$.

After $N_1*N_2$ rounds, procedure algorithm\_main will be executed. If there is no consumer with residual money, then the algorithm will exit.
Otherwise, the next round robin will occur.
\end{proof}

Let $T_{ob}, T_{ls}, T_{bd},$ and $T_{bs}$ denote the time taken for procedures outbid, lp\_solver, bal\_od and bal\_os, respectively. $T_{ls} = O(qm^2(m+l)L)$ and $T_{bd} + T_{bs} = mT_{ob} + nm$ where $T_{ob} = \displaystyle\log_{1+\epsilon}\left(\frac{e}{\epsilon}\right)^{|E|}\frac{ev_{max}}{\epsilon v_{min}}$ ($v_{max} = \max_{ij}x_{ij}(0)$ and $v_{min} = \min_{ij}x_{ij}(a_{max}))$. Let $|E|$ be the number of nonzero utilities.
Also, note that $\epsilon' = \displaystyle\min\{\epsilon^3/e,\max_{ij}v_{ij}(0)/(1+\epsilon) = v((1+\epsilon')\epsilon)$.

\begin{lemma}
\label{thm-terminate}
The time complexity of finding equilibrium in MEP is $O(N_0 * N_1 * N_2(T_{ls} + T_{bd} + T_{bs}))$.
\end{lemma}

\begin{proof}
Let us consider the worst case. For one round-robin, either procedure outbid or raise\_price occurs.
The number of iterations before a round robin occurs is bounded by $N_1 * N_2$ as shown in Claim \ref{claim2}.
Note that between round-robins, residual money decreases by a factor of ($1+\epsilon'$) of total money of consumers.
It means that total occurrence of round-robins is bounded by $\displaystyle N_0=\log_{1+\epsilon'}\frac{e}{\epsilon e_{min}}$.
Therefore, the total step is bounded by $N_0*N_1*N_2$.

Inside procedure prod\_reschedule, four events can occur in a call of outbid.
\begin{enumerate}
\item $y_{kj}$ becomes zero for some $k$. \label{outbid1}
\item $r_i$ becomes zero.\label{outbid2}
\item $\alpha_{ij}$ reduces by a factor of $(1+\epsilon)$.\label{outbid3}
\item $v_{ij}$ reaches $\alpha_{ij}$ in the inner while loop of algorithm main.\label{outbid4}
\end{enumerate}

The number of event (\ref{outbid1}) is bounded by the number of buyers having non-zero utilities on item $j$.
The total number of (\ref{outbid1}) events is bounded by $|E| \times N_1$.
The number of type (\ref{outbid2}) events is exactly equal to $n$ in every round of bidding.
At every type (\ref{outbid3}) event, bang-per-buck is reduced by a factor of ($1+\epsilon$) which varies from $\displaystyle\frac{v_{min}}{p_{max}}$ to $\displaystyle\frac{v_{max}}{p_{min}}$.
In every round only one event type of (\ref{outbid4}) occurs for each buyer.

Thus, outbid takes $\displaystyle|E|\log_{1+\epsilon}\frac{p_{max}}{p_{min}} + \log_{1+\epsilon}\frac{p_{max}v_{max}}{p_{min}v_{min}} = \displaystyle\log_{1+\epsilon}\left(\frac{e}{\epsilon}\right)^{|E|}\frac{ev_{max}}{\epsilon v_{min}}$ as shown in \cite{GKV04}.

Note that $lp\_solver$ is called per each iterations.
The function, $lp\_solver$, takes $O(qm^2(m+l)L)$, where $m \times l$ represents a matrix of production constraints and $L$ represents an input size.
$T_{lp\_solver} = O(qm^2(m+l)L)$.
It takes $mT_{ob}$ times for consumers to get money back and for the procedure decrease\_procedure, and it takes $nm$ times for producers to sell their produced items. $T_{bd} + T_{bs} = mT_{ob} + nm$.

Time complexity is $O(N_0 * N_1 * N_2(T_{ls} + T_{bd} + T_{bs}))$.
\end{proof}

\begin{theorem}
Approximation equilibrium in the market equilibrium with production, MEP, can be determined by a PTAS.
\end{theorem}
\begin{proof}
At termination, termination condition (\ref{spendmoney2}) is true because $r_i$ is low.
Other conditions are true by Lemma \ref{all-opt} which ensures the correctness.
The time complexity result follows from Lemma \ref{thm-terminate}.
\end{proof} 

\section{Conclusion}
\par In this paper, we show an auction-based algorithm for a production model where consumers have nonlinear utility functions and producers have a set of linear capacity constraints.
Our algorithm can also be extended to arbitrary convex production regions and the Arrow-Debreu model.

\appendix
\section{Appendix - Algorithm}
\begin{algorithm}[!h]
\caption{algorithm\_main}
\begin{algorithmic}[1]
\STATE initialize
\WHILE{there is extra demand, i.e. $\exists i: r_i > \epsilon e_i$}
\STATE satisfy\_demand($i$)
\STATE adjust\_bpb
\STATE prod\_reschedule
\ENDWHILE
\end{algorithmic}
\end{algorithm}

\begin{algorithm}[!h]
\caption{initialize}
\begin{algorithmic}[1]
\STATE $\forall j : p_j := \epsilon$
\STATE $\forall i : \alpha_i := \max_jv_{ij}(0)/p_j$
\STATE $\forall i : \mathcal{D}_i := \arg\max_jv_{ij}(0)/p_j$
\STATE $\forall s,j: z_{sj} := \epsilon/q$
\STATE $\forall i : x_{ij} := \epsilon$, where $j \in \mathcal{D}_i$
\STATE $\epsilon_1$ s.t. $\forall i,j: \min_{x_{ij}}v_{ij}(x_{ij})/(1+\epsilon) = v_{ij}((1+\epsilon_1)x_{ij})$
\STATE $\epsilon_2 := \epsilon^3/n, n = \sum_ie_i$
\STATE $\epsilon' := \min(\epsilon_1,\epsilon_2)$
\end{algorithmic}
\end{algorithm}

\begin{algorithm}[!h]
\caption{satisfy\_demand($i$)}
\begin{algorithmic}[1]
\STATE $\forall j: \alpha_{ij} := v_{ij}(x_{ij})/p_j$
\STATE $j:=\arg\max_l\alpha_{il}$
\IF{ $\exists k: y_{kj} > 0$ }
\STATE outbid($i, k, j,\alpha_{ij}/(1+\epsilon)$)
\STATE $\alpha_{ij} := v_{ij}(x_{ij})/p_j$
\ELSE
\STATE raise\_price($j$)
\ENDIF
\end{algorithmic}
\end{algorithm}

\begin{algorithm}[!h]
\caption{adjust\_bpb}
\begin{algorithmic}[1]
\WHILE{$\exists i: r_i > 0$ and $\exists j: \alpha_{ij}p_j < v_{ij}(x_{ij})$}
\IF{$\exists k: y_{kj} > 0$}
\STATE outbid($i,k,j,\alpha_{ij}$)
\ELSE
\STATE raise\_price($j$)
\ENDIF
\ENDWHILE
\end{algorithmic}
\end{algorithm}

\begin{algorithm}[!h]
\caption{outbid($i, j, k, \alpha$)}
\begin{algorithmic}[1]
\STATE $t_1 := y_{kj}$
\STATE $t_2 := r_i/p_j$
\IF {$(v_{ij}(a_j) \geq \alpha p_j)$}
\STATE $t_3 := a_j$
\ELSE
\STATE $t_3 := \min \delta: v_{ij}(x_{ij}+\delta) = \alpha p_j$
\ENDIF
\STATE $t := \min (t_1,t_2,t_3)$
\STATE $h_{ij} := h_{ij} + t$
\STATE $r_i := r_i - tp_j$
\STATE $y_{kj} := y_{kj} - t$
\STATE $r_k := r_k + tp_j/(1+\epsilon)$
\end{algorithmic}
\end{algorithm}

\begin{algorithm}[htbp]
\caption{raise\_price($j$)}
\begin{algorithmic}[1]
\STATE $p_j := p_j(1+\epsilon)$
\STATE $\forall i: y_{ij} := h_{ij}$
\STATE $\forall i: h_{ij} := 0$
\end{algorithmic}
\end{algorithm}

\begin{algorithm}[!h]
\caption{prod\_reschedule}
\begin{algorithmic}[1]
\WHILE {$true$}
\STATE get the optimal production plan, i.e. $\forall s : \hat{z}_{sj} := lp\_solver(p)$
\STATE $\forall s,j : z'_{sj} := z_{sj} \pm \epsilon'\sum_sz_{sj}/q$
\IF {$\sum_s\sum_jp_jz'_{sj} > \sum_s\sum_jp_jz_{sj}$}
\STATE bal\_od : $\forall j \in \mathcal{O}_d$
\STATE adjust\_bpb
\STATE bal\_os
\STATE check\_profit
\ELSE
\STATE break
\ENDIF
\ENDWHILE
\end{algorithmic}
\end{algorithm}

\ignore{
\begin{algorithm}[!h]
\caption{shrink\_model}
\begin{algorithmic}[1]
\STATE construct a bipartite graph $(U,V,E)$.
\STATE check whether a disjoint graph $(U',V',E')$ s.t. consumers have no demand on $V' \subseteq \mathcal{O}^-_d$
\STATE consider a new model on $(U \setminus U', V \setminus V', E \setminus E')$.
\end{algorithmic}
\end{algorithm}
}

\ignore{
\begin{algorithm}[!h]
\caption{prod\_outbid($j$)}
\begin{algorithmic}[1]
\WHILE {$\exists s : z_{sj} < 0 \bigwedge z'_{sj} < z_{sj}$}
\IF {$\exists k : y_{kj} > 0$}
\STATE $t := \min(y_{kj},z_{sj} - z'_{sj})$
\STATE $z_{sj} := z_{sj} - t$
\STATE $y_{kj} := y_{kj} - t$
\STATE $r_k := r_k + tp_j$
\ELSIF {$\exists s' : z_{s'j} < 0$}
\STATE $t := \min(-z_{s'j},z_{sj} - z'_{sj})$
\STATE $z_{s'j} := z_{s'j} + t$
\STATE $z_{sj} := z_{sj} - t$
\ELSE
\STATE raise\_price($j$)
\STATE continue
\ENDIF
\ENDWHILE
\end{algorithmic}
\end{algorithm}
}

\begin{algorithm}[!h]
\caption{bal\_od}
\begin{algorithmic}[1]
\WHILE {$\exists j \in \mathcal{O}^+_d$}
\IF {$\exists i : x_{ij} > 0$}
\STATE $t := \min(x_{ij} , \sum_ix_{ij}-\sum_sz_{sj})$
\STATE $x_{ij} := x_{ij} - t$
\STATE $r_i := r_i + tp_j$
\ENDIF
\ENDWHILE
\end{algorithmic}
\end{algorithm}

\ignore{
\begin{algorithm}[!h]
\caption{spend\_return\_money($i,j,r'_i$)}
let $x'$ be the amount of the previous allocation.\\
let $x$ be the amount of the current allocation.
\begin{algorithmic}[1]
\WHILE{ $r'_i > 0$ }
\IF{ $\exists k: y_{kj} > 0$ and $\alpha_{ij}p_j < v_{ij}(x_{ij})$}
\STATE $x'_{ij} := x_{ij}$
\STATE outbid($i, j, k, v_{ij}(x_{ij}+r'_i/p_j)$)
\STATE $r'_i := r'_i - p_j(x_{ij} - x'_{ij})$
\ELSE
\STATE raise\_price($j$)
\ENDIF
\ENDWHILE
\STATE $\alpha_{ij} := v_{ij}(x_{ij})/p_j$
\end{algorithmic}
\end{algorithm}
}

\begin{algorithm}[!h]
\caption{bal\_os}
\begin{algorithmic}[1]
\WHILE {$\forall j \in \mathcal{O}_s$}
\IF {$\exists i : j \in \mathcal{D}_i$ and $r_i > \epsilon e_i$}
\STATE purchase\_money($i,j,\sum_sz_{sj} - \sum_ix_{ij}$)
\ELSE
\IF {$\exists i : j \in \mathcal{D}_i$ and $r_i \leq \epsilon e_i$}
\STATE transfer\_money($i,j,\sum_sz_{sj} - \sum_ix_{ij}$)
\ENDIF
\IF {producers provide items at lower price, i.e. $\sum_sz_{sj} - \sum_ix_{ij} \leq \epsilon'\sum_ih_{ij}$}
\STATE sell\_lprice($i,j,\sum_sz_{sj} - \sum_ix_{ij}$)
\ELSIF {Not enough items at lower price, i.e. $\sum_sz_{sj} - \sum_ix_{ij} > \epsilon'\sum_ih_{ij}$}
\STATE decrease\_price($j$)
\ENDIF
\ENDIF
\ENDWHILE
\end{algorithmic}
\end{algorithm}

\begin{algorithm}[!h]
\caption{purchase\_money($i,j,t_o$)}
\begin{algorithmic}[1]
\STATE $t := \min(t_o,r_i/p_j)$
\STATE $t_o := t_o - t$
\STATE $h_{ij} := h_{ij}+t$
\STATE $r_i : = r_i - tp_j$
\end{algorithmic}
\end{algorithm}

\begin{algorithm}[!h]
\caption{transfer\_money($i,j,t_o$)}
\begin{algorithmic}[1]
\IF { $t_o > 0$ and $h_{ij'} > 0$ }
\STATE $t := \min(h_{ij'}p_j'/p_j,t_o)$
\STATE $h_{ij'} := h_{ij'} - tp_j/p_j'$
\STATE $h_{ij} := h_{ij} + t$
\STATE $t_o := t_o - t$
\ENDIF
\end{algorithmic}
\end{algorithm}

\begin{algorithm}[!h]
\caption{sell\_lprice($i,j,t_o$)}
\begin{algorithmic}[1]
\WHILE {$t_o > 0$}
\IF {$t_o \geq \epsilon' h_{ij}$}
\STATE $y_{ij} := y_{ij} + (1+\epsilon')h_{ij}$
\STATE $t_o := t_o-\epsilon' h_{ij}$
\STATE $h_{ij} := 0$
\ELSE
\STATE $y_{ij} := y_{ij} + t_o$
\STATE $h_{ij} := h_{ij} - t_o/(1+\epsilon')$
\STATE $t_o := 0$
\ENDIF
\ENDWHILE
\end{algorithmic}
\end{algorithm}

\begin{algorithm}[!th]
\caption{decrease\_price($j$)}
\begin{algorithmic}[1]
\ignore{
\STATE $t_1 := (1+\epsilon)^2\sum_ih_{ij}$
\STATE $t_2 := (1+\epsilon)\sum_iy_{ij}$
\STATE $\sum_ih_{ij} := \frac{t_1+t_2-\sum_sz_{sj}}{\epsilon}$
\STATE $\sum_iy_{ij} := \sum_sz_{sj} - \sum_ih_{ij}$
}
\STATE $\forall i: h_{ij} := (1+\epsilon)h_{ij} + y_{ij}$
\STATE $\forall i: y_{ij} := 0$
\STATE $p_j = p_j/(1+\epsilon)$
\STATE $\forall i : \alpha_i := \max_{j'} v_{ij'}/p_{j'}$
\STATE $\forall i : \mathcal{D}_i := \arg\max_{j'} v_{ij'}/p_{j'}$
\end{algorithmic}
\end{algorithm}

\begin{algorithm}[!th]
\caption{check\_profit}
let $z'$ and $p'$ be vectors of the previous iteration.
\begin{algorithmic}[1]
\IF {$\sum_s\sum_jz_{sj}p_j - \sum_jz'_{sj}p'_j < \gamma = \epsilon' \min_s\sum_jp_jz_{sj}$}
\STATE roll\_back
\STATE break
\ENDIF
\end{algorithmic}
\end{algorithm}

\clearpage

\end{document}